\documentclass[11pt]{article}

\usepackage{authblk}
\usepackage{amsthm}
\usepackage[letterpaper, left=1in, right=1in, bottom=1in, top=1in]{geometry}
\usepackage{graphicx}
\usepackage{amsmath,amssymb}
\usepackage{booktabs}
\usepackage{multirow}

\usepackage[T1]{fontenc}
\usepackage{diagbox}
\usepackage{url}
\usepackage{hyperref}
\usepackage{cleveref}

\newtheorem{theorem}{Theorem}
\newtheorem{corollary}{Corollary}
\newtheorem{lemma}{Lemma}

\newtheorem{proposition}{Proposition}

\theoremstyle{definition}
\newtheorem{definition}{Definition}

\newcommand{\lcp}{\mathit{lcp}}
\newcommand{\rev}[1]{{#1}^{\mathit{R}}}

\newcommand{\LPal}{\mathsf{LPal}}
\newcommand{\MPal}{\mathsf{MPal}}
\newcommand{\CT}{\mathsf{CT}}
\newcommand{\PD}{\mathsf{PD}}

\begin{document}
\title{Computing maximal palindromes in non-standard matching models}

\author[1]{Takuya~Mieno}
\author[2]{Mitsuru~Funakoshi\thanks{Current affiliation: NTT Communication Science Laboratories, Japan.}}
\author[2]{Yuto~Nakashima}
\author[2]{Shunsuke Inenaga}
\author[3]{Hideo~Bannai}
\author[2]{Masayuki~Takeda}

\affil[1]{The University of Electro-Communications, Japan}
\affil[2]{Kyushu University, Japan}
\affil[3]{Institute of Science Tokyo, Japan}

\date{}
\maketitle

\begin{abstract}
  \emph{Palindromes} are popular and important objects in textual data processing,
  bioinformatics, and combinatorics on words.
  Let $S = XaY$ be a string
  where $X$ and $Y$ are of the same length, and $a$ is either a single character or the empty string.
  Then, there exist two alternative definitions for palindromes:
  $S$ is said to be a palindrome if $S$ is equal to its reversal $\rev{S}$ (Reversal-based definition);
  or if its right-arm $Y$ is equal to the reversal of its left-arm $\rev{X}$ (Symmetry-based definition).
  It is clear that if the ``equality''~($\approx$) used in both definitions
  is exact character matching~($=$),
  then the two definitions are the same.
  However, if we apply other string-equality criteria $\approx$,
  including the \emph{complementary-matching model} for biological sequences,
  the \emph{Cartesian-tree model} [Park et al., TCS 2020],
  the \emph{parameterized model} [Baker, JCSS 1996],
  the \emph{order-preserving model} [Kim et al., TCS 2014],
  and the \emph{palindromic-structure model} [I et al., TCS 2013],
  then are the reversal-based palindromes and the symmetry-based palindromes the same?
  To the best of our knowledge,
  no previous work has considered or answered this natural question.
  In this paper, we first provide answers to this question,
  and then present efficient algorithms for computing all
  \emph{maximal palindromes under the non-standard matching models} in a given string.
  After confirming that Gusfield's offline suffix-tree-based algorithm for computing
  maximal symmetry-based palindromes can be readily extended to the aforementioned matching models,
  we show how to extend Manacher's online algorithm for computing maximal reversal-based palindromes
  in linear time for all the aforementioned matching models.
\end{abstract}

\section{Introduction}\label{sec:introduction}
Finding characteristic patterns in strings,
such as {tandem repeat}, {unique factor}, and {palindrome},
is a fundamental task in string data processing.
Of course, what kind of characteristic string we want depends on the application domain.
For example, when strings represent DNA/RNA sequences,
tandem repeats can help find genetic diseases~\cite{gusfield97:_algor_strin_trees_sequen},
unique factors may be helpful in preprocessing PCR primer design~\cite{Pei}, and
(gapped) palindromes may represent secondary structures of RNA called hairpin~\cite{gusfield97:_algor_strin_trees_sequen}.
Let us consider the case
when a string represents a stock chart, which is a numerical sequence.
In this case,
one is less interested in their exact values (i.e., exact prices) than
in finding price fluctuations from the stock chart.
Motivated by discovering such fluctuations, Park et al.~\cite{ParkBALP20} introduced a notion of \emph{Cartesian-tree match}.
The \emph{Cartesian-tree} of a numeric string $S$ of length $n$ is an ordered binary tree recursively defined below:
the root is the leftmost occurrence $i$ of the smallest value in $S$,
the left child of the root is the Cartesian-tree of $S[1..i-1]$, and
the right child of the root is the Cartesian-tree of $S[i+1.. n]$.
We say that two strings of equal length Cartesian-tree match if their Cartesian-trees are isomorphic.
Cartesian-tree matching can capture shapes of stock chart patterns such as \emph{head-and-shoulder}~\cite{ParkBALP20}.
Also, some significant stock chart patterns exhibit \emph{palindrome-like} structures (e.g., head-and-shoulder and double-top).
Hence, enumerating such palindrome-like structures in a stock chart
is expected to improve the efficiency of chart pattern searches.
The first motivation of our research is
to explore efficient methods to compute such palindrome-like structures by applying Cartesian-tree matching.

First, let us revisit the definition of palindromes in the context of the standard matching model,
that is, the ``equality''~($\approx$) between two strings
is the exact matching on a character-by-character basis~($=$).
There are two common definitions for palindromes:
Let $S = XaY$ be a string,
where $X$ and $Y$ are the same length, and $a$ is either a single character or the empty string.
Then, $S$ is said to be a palindrome if:
\begin{description}
  \item[Reverse-based definition:] $S$ is equal to its reversal $\rev{S}$;
  \item[Symmetry-based definition:] $Y$ is equal to the reversal of $X$.
\end{description}
It is clear that if we assume the standard matching model,
then the two definitions are the same,
meaning that a string $S$ is a reverse palindrome
iff $S$ is a symmetric palindrome.
Observe this with examples such as $\mathtt{racecar}$ and $\mathtt{noon}$.
However, the two definitions above differ
if we apply the Cartesian-tree model.
For instance,
the Cartesian-tree of $S_1 = \mathtt{223}$
is the tree with no branches that slopes down to the right
while the Cartesian-tree of $S_1^R = \mathtt{322}$ is shaped like a \emph{logical AND symbol} ($\land$)
because the leftmost smallest element in $S_1$ is the second one.
Thus, $S_1$ is not a reverse palindrome.
On the other hand,
$S_1 = \mathtt{223}$ is a symmetric palindrome
since the Cartesian-trees of its right-half $\mathtt{3}$ and the reversal of its left-half $(\mathtt{2})^R = \mathtt{2}$ are both singletons.
In addition,
the two definitions are also different
if we apply the \emph{complementary-matching model}
where \texttt{A} matches \texttt{T}, and \texttt{G} matches \texttt{C},
which is a standard matching model for DNA sequences.
For instance, a string $S_2 = \mathtt{AGTCT}$
is not a reverse palindrome since $S_2[3] = \mathtt{T}$ and $S_2^R[3] = \mathtt{T}$ do not complementary-match,
but is a symmetric palindrome
since $(\mathtt{AG})^R = \mathtt{GA}$ and $\mathtt{CT}$ complementary-match.
Now we raise the following question:
If we apply other string-equality criteria $\approx$,
including
the \emph{parameterized model}~\cite{Baker96},
the \emph{order-preserving model}~\cite{KimEFHIPPT14},
and the \emph{palindromic-structure model}~\cite{IIT13},
then are the reverse palindromes and the symmetric palindromes the same?
This question is interesting because,
while the symmetry-based palindrome definition requires only $\rev{X} \approx Y$,
the reversal-based palindrome definition requires more strict matching
under $\approx$ so that
$S \approx \rev{S} \Leftrightarrow XaY \approx \rev{(XaY)} \Leftrightarrow XaY \approx \rev{Y} a \rev{X}$.
Thus, these two types of palindromes can be quite different for some matching criterion $\approx$ (see also Figure~\ref{fig:scsr-palindromes}).
\begin{figure}[tbh]
  \centerline{
    \includegraphics[width=1.0\textwidth]{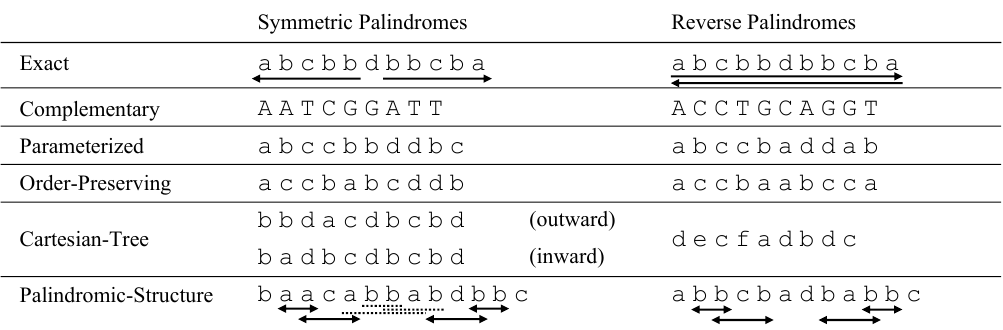}
  }
  \caption{
    This figure shows examples of a palindrome in each non-standard matching model.
    A bijection $f$ such that $f(\mathtt{a})=\mathtt{c}$, $f(\mathtt{b})=\mathtt{b}$, $f(\mathtt{c})=\mathtt{d}$, $f(\mathtt{d})=\mathtt{a}$
    gives the parameterized sym-palindrome.
    A bijection $g$ such that $g(\mathtt{a})=\mathtt{b}$, $g(\mathtt{b})=\mathtt{a}$, $g(\mathtt{c})=\mathtt{d}$, $g(\mathtt{d})=\mathtt{c}$
    gives the parameterized rev-palindrome.
    In the palindromic-structure sym-palindrome,
    though palindromes $\mathtt{bb}$, $\mathtt{abba}$, and $\mathtt{bab}$ exist,
    these palindromes are ignored in this symmetric condition.
  }
  \label{fig:scsr-palindromes}
\end{figure}
To our knowledge,
no previous work has considered or answered this natural question.
In this paper, we first provide quick answers to this question
and present efficient algorithms for computing such palindromes in a given string under the non-standard matching models.

One of the well-studied topics regarding palindromes is \emph{maximal palindromes},
which are substring palindromes whose left-right extensions are not palindromes.
It is interesting and important to find all maximal palindromes
in a given string $T$ because any substring palindrome of $T$ can be obtained by
removing an equal number of characters from the left and the right of some maximal palindrome.
Hence, by computing all maximal palindromes of a string, we obtain a compact representation of all palindromes in the string.
Manacher~\cite{Manacher75} showed an elegant algorithm for finding all maximal palindromes in a given string $T$.
Manacher's algorithm works in an online manner (processes the characters of $T$ from left to right) and runs in $O(n)$ time and space for general (unordered) alphabets,
where $n$ is the length of the input string $T$.
Later, Gusfield~\cite{gusfield97:_algor_strin_trees_sequen} showed
another famous algorithm for computing all maximal palindromes.
Gusfield's algorithm uses
the suffix tree~\cite{Weiner73} built on a concatenation of $T$ and its reversal $\rev{T}$ that is enhanced with an LCA (lowest common ancestor) data structure~\cite{SchieberV88}.
Gusfield's method is offline (processes all the characters of $T$ together)
and works in $O(n)$ time and space for linearly-sortable alphabets,
including integer alphabets
of size $\mathsf{poly}(n)$.
We remark that Gusfield's algorithm uses only the symmetry-based definition
for computing maximal palindromes.
That is, computing the longest common prefix of $T[c..n]$ and $\rev{(T[1..c-1])}$
for each integer position $c$ in $T$ gives us the maximal palindrome
of even length centered at $c-0.5$.
Maximal palindromes of odd lengths can be computed analogously.
On the other hand, Manacher's algorithm
(which will be briefly recalled in a subsequent section)
uses both reversal-based and symmetry-based definitions for computing maximal palindromes.

\begin{table}[htb]
  \begin{center}
    \begin{tabular}{|c||c|c|} \hline
      \diagbox{Matching}{Type}               & Symmetric Palindromes                                               & Reverse Palindromes           \\ \hline \hline
      \multirow{2}{*}{Exact}                 & $O(n)$ time~\cite{gusfield97:_algor_strin_trees_sequen}             & $O(n)$ time~\cite{Manacher75} \\
                                             & (linearly sortable)                                                 & (general unordered)           \\ \hline
      \multirow{2}{*}{Complementary}         & $O(n)$ time                                                         & $O(n)$ time                   \\
                                             & (linearly sortable)                                                 & (general unordered)           \\  \hline
      \multirow{2}{*}{Cartesian-Tree}        & $O(n \log n)$ time                                                  & $O(n)$ time                   \\
                                             & (general ordered)                                                   & (linearly sortable)           \\\hline
      \multirow{2}{*}{Parameterized}         & $O(n \log (\sigma+\pi))$ time                                       & $O(n)$ time                   \\
                                             & (linearly sortable)                                                 & (linearly sortable)           \\ \hline
      \multirow{2}{*}{Order-Preserving}      & $O(n \log \log^2 n / \log \log \log n)$ time                        & $O(n)$ time                   \\
                                             & (linearly sortable)                                                 & (general ordered)           \\  \hline
      \multirow{2}{*}{Palindromic-Structure} & $O(n \min\{ \sqrt{\log n}, \log \sigma / \log \log \sigma \})$ time & $O(n)$ time                   \\
                                             & (general unordered)                                                 & (general unordered)           \\\hline
    \end{tabular}
  \end{center}
  \caption{Time complexities of algorithms for computing each type of palindromes, where $n$ denotes the length of the input string, $\sigma$ is the (static) alphabet size, and $\pi$ is the size of the parameterized alphabet for parameterized matching. The time complexities are valid under some assumptions for the alphabet, which are designated in parentheses. Each of the algorithms uses $O(n)$ space.}
  \label{tab:complexity}
\end{table}

In this paper, we propose a new framework for formalizing palindromes under the non-standard matching models, which we call \emph{Substring Consistent Symmetric and Two-Transitive Relations} (\emph{SCSTTRs}) that include the complementary-matching model and \emph{Substring Consistent Equivalence Relations} (\emph{SCERs})~\cite{MatsuokaAIBT16}.
We note that SCERs include
the Cartesian-tree model~\cite{ParkBALP20},
the parameterized model~\cite{Baker96},
the order-preserving model~\cite{KimEFHIPPT14},
and the palindromic-structure model~\cite{IIT13}.
As far as we are aware, the existing algorithms are designed only for computing standard maximal palindromes based on exact character matching, and it is not clear how they can be adapted for the aforementioned palindromes under the non-standard matching models.
Our first claim is that Gusfield's framework is easily extensible for palindromes under the non-standard matching models: if one has a suffix-tree-like data structure that is built on a suitable encoding for a given matching criterion $\approx$ enhanced with the LCA data structure, then one can readily compute all maximal symmetric palindromes under the non-standard matching models in $O(n)$ time
(the details are shown in Section~\ref{sec:algo_symmetric}).
Thus, the construction time for the suffix-tree-like data structure dominates the total time complexity (see Table~\ref{tab:complexity}).
On the other hand, extending Manacher's algorithm for palindromes under the non-standard matching models is much more involved.
The main contribution of this paper is to design Manacher-style algorithms that compute all maximal reverse palindromes in $O(n)$ time
under \emph{all} the non-standard matching models considered in this paper (see Table~\ref{tab:complexity}).
We remark that a straightforward implementation of the extension of Manacher's algorithm requires quadratic time.
We then reduce the time complexity to (near) linear by exploiting and utilizing combinatorial properties of the non-standard matching models (in particular, the Cartesian-tree model, the parameterized model, and the palindromic-structure model require non-trivial elaborations).

This paper is a full version of a preliminary version~\cite{FunakoshiMNIBT24}.
This paper includes complete proofs of theorems and possible future work, which were omitted in the preliminary version.

\paragraph*{Related work}
The Cartesian-tree matching problem is introduced by Park et al.~\cite{ParkBALP20}.
They showed a linear-time algorithm to solve the problem
and also proposed an $O(n)$-space index for the problem,
which can be constructed in $O(n\log n)$ time.
Kim and Cho~\cite{KimC21} developed a compact $3n+o(n)$-\emph{bits} index.
Nishimoto et al.,~\cite{NishimotoFNI21} proposed another $O(n)$-space index,
which can be constructed in $O(n\log\sigma)$ time where $\sigma$ is the alphabet size.
A probabilistic method for multiple pattern Cartesian-tree matching is also studied~\cite{SongGRFLP21}.
Many other variants of the Cartesian-tree matching problem are studied,
including
matching on indeterminate strings~\cite{GawrychowskiGL20},
subsequence matching~\cite{OizumiKMIA22},
approximate matching~\cite{AuvrayDGL23}, and
the longest common factor computing~\cite{FaroLPS23}.

As for other or more general settings,
Kikuchi et al.~\cite{KikuchiHYS20} considered the problem of finding \emph{covers} under SCERs.
Kociumaka et al.~\cite{KociumakaRRW16} introduced \emph{squares} (tandem repeats) in the non-standard matching models, including order-preserving matching.
Gawrychowski et al.~\cite{OrderPreservingSquares} presented a worst-case optimal time algorithm to compute the distinct order-preserving matching squares in a string.
We note that our results are not readily obtained by the results above and require new techniques due to the different natures between covers/squares and palindromes.

\section{Preliminaries}\label{sec:preliminaries}
\subsection{Notations}
Let $\Sigma$ be an \emph{alphabet}.

Alphabet $\Sigma$ is said to be ordered
if there is a total order, denoted as $\preccurlyeq$, on $\Sigma$.
Otherwise, $\Sigma$ is said to be unordered.
For any two elements $a$ and $b$ of an ordered alphabet $\Sigma$,
we write $a \prec b$ if $a \preccurlyeq b$ and $a \ne b$.

An element of $\Sigma^*$ is called a \emph{string}.
The length of a string $T$ is denoted by $|T|$.
The empty string $\varepsilon$ is the string of length $0$,
namely, $|\varepsilon| = 0$.
For a string $T = xyz$, $x$, $y$ and $z$ are called
a \emph{prefix}, \emph{substring}, and \emph{suffix} of $T$, respectively.
For a string $T$ and an integer $1 \leq i \leq |T|$,
$T[i]$ denotes the $i$-th character of $T$,
and for two integers $1 \leq i \leq j \leq |T|$,
$T[i..j]$ denotes the substring of $T$
that begins at position $i$ and ends at position $j$.
For convenience, let $T[i..j] = \varepsilon$ when $i > j$.
The reversal of a string $T$ is denoted by $\rev{T}$,
i.e., $\rev{T} = T[|T|] \cdots T[1]$.
For two strings $X$ and $Y$,
$\lcp(X, Y)$ denotes the length of the longest common prefix of $X$ and $Y$.
Namely, $\lcp(X, Y) = \max\{\ell \mid X[1.. \ell] = Y[1.. \ell]\}$.
A \emph{rightward longest common extension} (\emph{rightward LCE})
query on a string $T$
is to compute $\lcp(T[i..|T|], T[j..|T|])$
for given two positions $1 \leq i < j \leq |T|$.
Similarly, a \emph{leftward LCE} query is
to compute $\lcp(\rev{T[1..i]}, \rev{T[1..j]})$.
Then, an \emph{outward LCE} (resp. \emph{inward LCE}) query is, given two positions $1 \leq i < j \leq |T|$,
to compute $\lcp(\rev{T[1..i]}, T[j..|T|])$ (resp. $\lcp(T[i..|T|], \rev{T[1..j]})$).

A string $T$ is called a \emph{palindrome} if $T = \rev{T}$.
We remark that the empty string $\varepsilon$ is also
considered to be a palindrome.
A non-empty substring palindrome $T[i..j]$
is said to be a \emph{maximal palindrome} in $T$
if $T[i-1] \neq T[j+1]$, $i = 1$, or $j = |T|$.
For any non-empty substring palindrome $T[i..j]$ in $T$,
$\frac{i+j}{2}$ is called its \emph{center}.
It is clear that for each center $c = 1, 1.5, \ldots, n-0.5, n$,
we can identify the maximal palindrome $T[i..j]$ whose center is $c$
(namely, $c = \frac{i+j}{2}$).
Thus, there are exactly $2n-1$ maximal palindromes in a string of length $n$
(including empty strings, which occur at non-integer centers $c$ when $T[c-0.5] \neq T[c+0.5]$).
Manacher~\cite{Manacher75} showed an online algorithm
that computes all maximal palindromes in a string $T$ of length $n$ in $O(n)$ time.
An alternative offline approach is to use outward LCE queries
for $2n-1$ pairs of positions in $T$.
Using the suffix tree~\cite{Weiner73} for string $T\$\rev{T}\#$
enhanced with a lowest common ancestor data structure~\cite{HarelT84,SchieberV88,BenderF00},
each outward LCE query can be answered in $O(1)$ time
where $\$$ and $\#$ are special characters that do not appear in $T$.
Preprocessing for this approach takes $O(n)$ time and space~\cite{Farach-ColtonFM00,gusfield97:_algor_strin_trees_sequen}
when alphabet $\Sigma$ is linearly-sortable.

\subsection{Substring Consistent Symmetric and Two-Transitive Relation}
A \emph{Substring Consistent Equivalence Relation} (\emph{SCER}) is an equivalence relation $\approx_\mathrm{SCER}$ such that $X \approx_\mathrm{SCER} Y$ for two strings $X, Y$ of equal length means that $X[i..j] \approx_\mathrm{SCER} Y[i..j]$ holds for all $1 \leq i \leq j \leq |X| = |Y|$.
Matsuoka et al.~\cite{MatsuokaAIBT16} defined the notion and considered the pattern matching problems with SCER.
Afterward, several problems over SCER have been considered (e.g.,~\cite{Hendrian20,KikuchiHYS20,JargalsaikhanHY22}).

In this paper, we treat a more general relation; Substring Consistent Symmetric and Two-Transitive Relation (SCSTTR).
An SCSTTR is a relation $\approx_\mathrm{SCSTTR}$ that satisfies the following conditions for strings
$W, X, Y, Z$ of equal length:
\begin{enumerate}
  \item Satisfies the symmetric law; namely, if $X \approx_\mathrm{SCSTTR} Y$, then $Y \approx_\mathrm{SCSTTR} X$ also holds.
  \item Satisfies the two-transitive law; namely, if $W \approx_\mathrm{SCSTTR} X$, $X \approx_\mathrm{SCSTTR} Y$, and $Y \approx_\mathrm{SCSTTR} Z$, then $W \approx_\mathrm{SCSTTR} Z$ also holds.
  \item Satisfies the substring consistency; namely, if $X \approx_\mathrm{SCSTTR} Y$, then
    $X[i..j] \approx_\mathrm{SCSTTR} Y[i..j]$ holds for all $1 \leq i \leq j \leq |X| = |Y|$.
\end{enumerate}

We remark that any SCER is also an SCSTTR.
Our aim for introducing SCSTTRs is to deal with the complementary-matching model including the Watson-Crick (WK) model~\cite{WatsonC1953} for $\Sigma = \{\mathsf{A}, \mathsf{C}, \mathsf{G}, \mathsf{T}\}$ that is \emph{not} an SCER.
For each of the SCSTTR matching models below, encodings with which the matching can be reduced to exact matching are known.

\paragraph*{Complementary Matching~\cite{LucaL06,KariM07}}
Let $f$ be a function on $\Sigma$ that has the two following properties: (1) $f(uv) = f(u)f(v)$ for all strings $u, v \in \Sigma^*$, and (2) $f^2$ equals the identity.
Two strings $X$ and $Y$ are said to complementary-match if there is a function $f$ satisfying $X = f(Y)$.
For example, if $f(\mathsf{A}) = \mathsf{T}$, $f(\mathsf{C}) = \mathsf{G}$, $f(\mathsf{T}) = \mathsf{A}$, and $f(\mathsf{G}) = \mathsf{C}$ on $\Sigma=\{\mathsf{A, T, G, C}\}$, then $X=\mathsf{AGCTAT}$ and $Y=\mathsf{TCGATA}$ complementary-match.
Reversal-based palindromes under complementary-matching have been studied as $\theta$-palindromes~\cite{LucaL06,KariM07}
\footnote{$\theta$-palindromes are defined by the function $\theta$ that is not morphic but is antimorphic in the literature, while we use the morphic function $f$ since they can be treated as the same in our arguments.}.
In this paper, we consider both the reverse and symmetric definitions and call them Theta rev-palindromes and Theta sym-palindromes, respectively.

\paragraph*{Cartesian-tree Matching~\cite{ParkBALP20}}
The Cartesian-tree $\CT(T)$ of a string $T$ is the ordered binary tree recursively defined as follows~\cite{Vuillemin80}:
\begin{enumerate}
  \item If $T$ is the empty string, $\CT(T)$ is empty.
  \item If $T$ is not empty, let $T[i]$ be the leftmost occurrence of the smallest character in $T$.
    Then, the root of $\CT(T)$ is $T[i]$, the left-side subtree of $\CT(T)$ is $\CT(T[1..i-1])$, and the right-side subtree of $\CT(T)$ is $\CT(T[i+1..|T|])$.
\end{enumerate}
We say that two strings $X$ and $Y$ Cartesian-tree match and denote it by $X \approx_{\mathrm{ct}} Y$
if $\CT(X)$ and $\CT(Y)$ are isomorphic as unlabeled ordered trees.
For example, $X \approx_{\mathrm{ct}} Y$ holds for two strings $X = \mathsf{cabdcf}$ and $Y = \mathsf{eaacbc}$ (see Figure~\ref{fig:CTmatch}).
\begin{figure}[t!]
  \centerline{
    \includegraphics[width=0.6\linewidth]{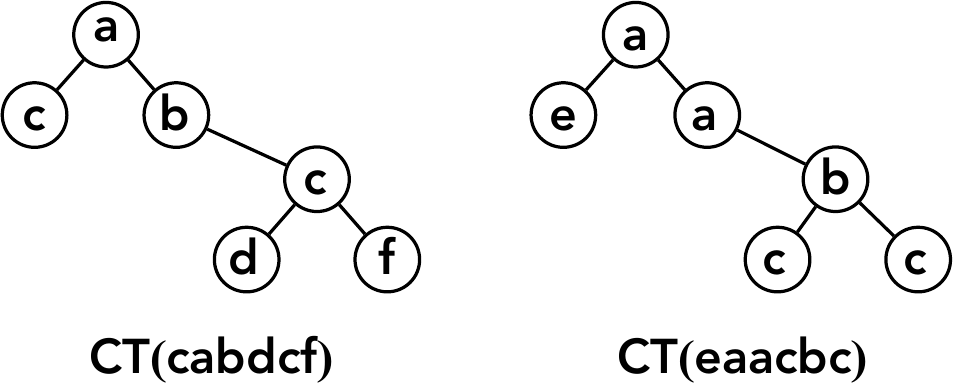}
  }
  \caption{
    Illustration for Cartesian-trees of $X = \mathsf{cabdcf}$ and $Y = \mathsf{eaacbc}$.
    Since they are isomorphic except their node labels, $X \approx_{\mathrm{ct}} Y$ holds.
  }
  \label{fig:CTmatch}
\end{figure}

The \emph{parent distance encoding} $\PD_T$ of a string $T$ is an integer sequence of length $|T|$ such that:
$$ \PD_T[i] =
\begin{cases}
  i - \max_{1 \leq j < i} \{ j : T[j] \preccurlyeq T[i]\} & \text{if such $j$ exists}; \\
  0                                                       & \text{otherwise}.
\end{cases}
$$
For any two strings $X$ and $Y$, $X \approx_{\mathrm{ct}} Y$ iff $\PD_X = \PD_Y$~\cite{ParkBALP20}.
For example, again consider two strings $X = \mathsf{cabdcf}$ and $Y = \mathsf{eaacbc}$, which satisfy $X \approx_{\mathrm{ct}} Y$,
then $\PD_{X= \mathsf{cabdcf}}=001121 = \PD_{Y= \mathsf{eaacbc}}$.
We note that if $S$ is a prefix of $T$, then the number of $0$s in $\PD_T$ is at least the number of $0$s in $\PD_S$.

\paragraph*{Parameterized Matching~\cite{Baker93}}
Let $\Sigma$ and $\Pi$ be disjoint sets of characters, respectively, called a static and parameterized alphabet.
Two strings $X$ and $Y$ are said to parameterized match if there is a renaming bijection over the alphabet $\Pi$ that transforms $X$ into $Y$.
We write $X \approx_{\mathrm{para}} Y$ iff strings $X$ and $Y$ parameterized match.

The \emph{previous encoding} $\mathsf{PE}_T$ of a string $T$ is a sequence of length $|T|$ such that:
$$ \mathsf{PE}_T[i] =
\begin{cases}
  T[i]                                         & \text{if $T[i] \in \Sigma$};                  \\
  i - \max_{1 \leq j < i} \{ j : T[j] = T[i]\} & \text{if $T[i] \in \Pi$ and such $j$ exists}; \\
  0                                            & \text{otherwise}.
\end{cases} $$

For any two strings $X$ and $Y$, $X \approx_{\mathrm{para}} Y$ iff $\mathsf{PE}_X = \mathsf{PE}_Y$~\cite{Baker93}.
For example, again consider two strings $X=\mathsf{aabaCbC}$ and $Y=\mathsf{ddadCaC}$, which satisfy $X \approx_{\mathrm{para}} Y$,
then $\mathsf{PE}_{X=\mathsf{aabaCbC}} = 0102\mathsf{C}3\mathsf{C} = \mathsf{PE}_{Y=\mathsf{ddadCaC}}$.
For simplicity, we assume in the rest of this paper that the string consists only of characters from $\Pi$.
Only trivial modifications to our algorithms are required when considering characters from $\Sigma$ as well.

\paragraph*{Order-preserving Matching~\cite{KimEFHIPPT14}}
For ordered alphabets,
two strings $X$ and $Y$ are said to order-preserving match if the relative orders of $X$ correspond to those of $Y$.
Namely, $X[i] \preccurlyeq X[j] \Leftrightarrow Y[i]\preccurlyeq Y[j]$ for all $1 \leq i, j \leq |X| = |Y|$.
We write $X \approx_{\mathrm{op}} Y$ iff strings $X$ and $Y$ order-preserving match.

Let $\alpha_T[i]$ (resp. $\beta_T[i]$) be the rightmost occurrence of the predecessor (resp. the successor) of $T[i]$ in $T[1..i-1]$.
Namely,
$ \alpha_T[i] = \max(\{ j < i : T[j]$ is the largest element satisfying $T[j] \preccurlyeq T[i]\}\cup\{0\})$
and
$\beta_T[i] = \max(\{ j < i: T[j]$ is the smallest element satisfying $T[j] \succcurlyeq T[i] \}\cup\{0\})$.
Let $\mathsf{Code}_T$ be a pairwise sequence such that $\mathsf{Code}_T = (\alpha_T[1], \beta_T[1]), \ldots, (\alpha_T[|T|], \beta_T[|T|])$ holds.
For any two strings $X$ and $Y$, $X \approx_{\mathrm{op}} Y$ iff $\mathsf{Code}_X = \mathsf{Code}_Y$ holds~\cite{CrochemoreIKKLP16}.
For example, again consider two strings $X = \mathsf{cecag}$ and $Y = \mathsf{hohbr}$, which satisfy $X \approx_{\mathrm{op}} Y$,
then $\mathsf{Code}_{X=\mathsf{cecag}}=(0,0)(1,0)(1,1)(0,3)(2,0) = \mathsf{Code}_{Y=\mathsf{hohbr}}$.

\paragraph*{Palindromic-structure Matching~\cite{IIT13}}
Two strings $X$ and $Y$ are said to palindromic-structure match if the length of the maximal palindrome at each center position in $X$ is equal to that of $Y$.
We write $X \approx_{\mathrm{pal}} Y$ iff strings $X$ and $Y$ palindromic-structure match.

Let $\LPal_T$ be the integer sequence of length $|T|$ such that $\LPal_T[i]$ stores the length of the longest suffix palindrome of the position $i$ in $T$.
For any two strings $X$ and $Y$, $X \approx_{\mathrm{pal}} Y$ iff $\LPal_X$ is equal to $\LPal_Y$~\cite{IIT13}.
For example, again consider two string $X = \mathsf{aabacdca}$ and $Y = \mathsf{ccacdadc}$, which satisfy $X \approx_{\mathrm{pal}} Y$,
then $\LPal_{\mathsf{aabacdca}}=12131135= \LPal_{\mathsf{ccacdadc}}$.

\subsection{Our Problems}
In this paper, we consider the following two definitions of palindromes:

\begin{definition}[SCSTTR symmetry-based palindromes]
  A string $P$ is called an SCSTTR sym-palindrome
  if $\rev{(P[1..\lfloor{|P|/2}\rfloor])} \approx_{\mathrm{SCSTTR}} P[\lceil{|P|/2}\rceil..|P|]$ holds.
\end{definition}

\begin{definition}[SCSTTR reversal-based palindromes]
A string $P$ is called an SCSTTR rev-{\linebreak}palindrome
  if $P \approx_{\mathrm{SCSTTR}} \rev{P}$ holds.
\end{definition}

As for palindromes under the exact matching, the above two definitions are the same.
However, as for palindromes under the non-standard matching models, these definitions are not equivalent.
For example, $\mathsf{ATTGAAT}$ is not a WK rev-palindrome but is a WK sym-palindrome.
Also, $\mathsf{CACB}$ is not a parameterized rev-palindrome but is a parameterized sym-palindrome.
We give examples of SCSTTR sym-/rev-palindromes in Figure~\ref{fig:scsr-palindromes}.
This paper considers problems for computing maximal SCSTTR (complementary, parameterized, order-preserving, Cartesian-tree, and palindromic-structure) palindromes for the two definitions.
Firstly, we notice that SCSTTR rev-palindromes have symmetricity:
\begin{lemma}\label{lem:symmetricity_of_SCSTTR_revpal}
  Let $P$ be an SCSTTR rev-palindrome.
  If a substring $P[i.. j]$ of $P$ is an SCSTTR rev-palindrome,
  then the substring at the symmetrical position is also an SCSTTR rev-palindrome,
  namely, $P[|P|-j+1.. |P|-i+1]$ is an SCSTTR rev-palindrome.
\end{lemma}
\begin{proof}
  Since $P$ is an SCSTTR rev-palindrome,
  $P[i.. j] \approx_{\mathrm{SCSTTR}} P[|P|-j+1.. |P|-i+1]^R$ and
  $P[i.. j]^R \approx_{\mathrm{SCSTTR}} P[|P|-j+1.. |P|-i+1]$ hold.
  Also, since $P[i.. j]$ is an SCSTTR rev-palindrome,
  $P[i.. j] \approx_{\mathrm{SCSTTR}} P[i.. j]^R$ holds.
  Combining these three equations under SCSTTR, we obtain
  $P[|P|-j+1.. |P|-i+1] \approx_{\mathrm{SCSTTR}} P[|P|-j+1.. |P|-i+1]^R$ from the symmetric law and two-transitive law.
  Therefore, $P[|P|-j+1.. |P|-i+1]$ is an SCSTTR rev-palindrome.
\end{proof}
This lemma allows us to design a Manacher-like algorithm for computing maximal SCSTTR rev-palindromes in Section~\ref{sec:algo_reverse}.

Next, we notice that the Cartesian-tree matching is not closed under reversal;
namely, $X\approx_{\mathsf{ct}} Y$ does not imply $\rev{X} \approx_{\mathsf{ct}} \rev{Y}$.
For instance, $\mathtt{babaab} \approx_{\mathsf{ct}} \mathtt{bababb}$ but $\mathtt{baabab} \not \approx_{\mathsf{ct}} \mathtt{bbabab}$~(see also Figure~\ref{fig:reversalCT}).
\begin{figure}[tb]
  \centerline{\includegraphics[width=0.6\textwidth]{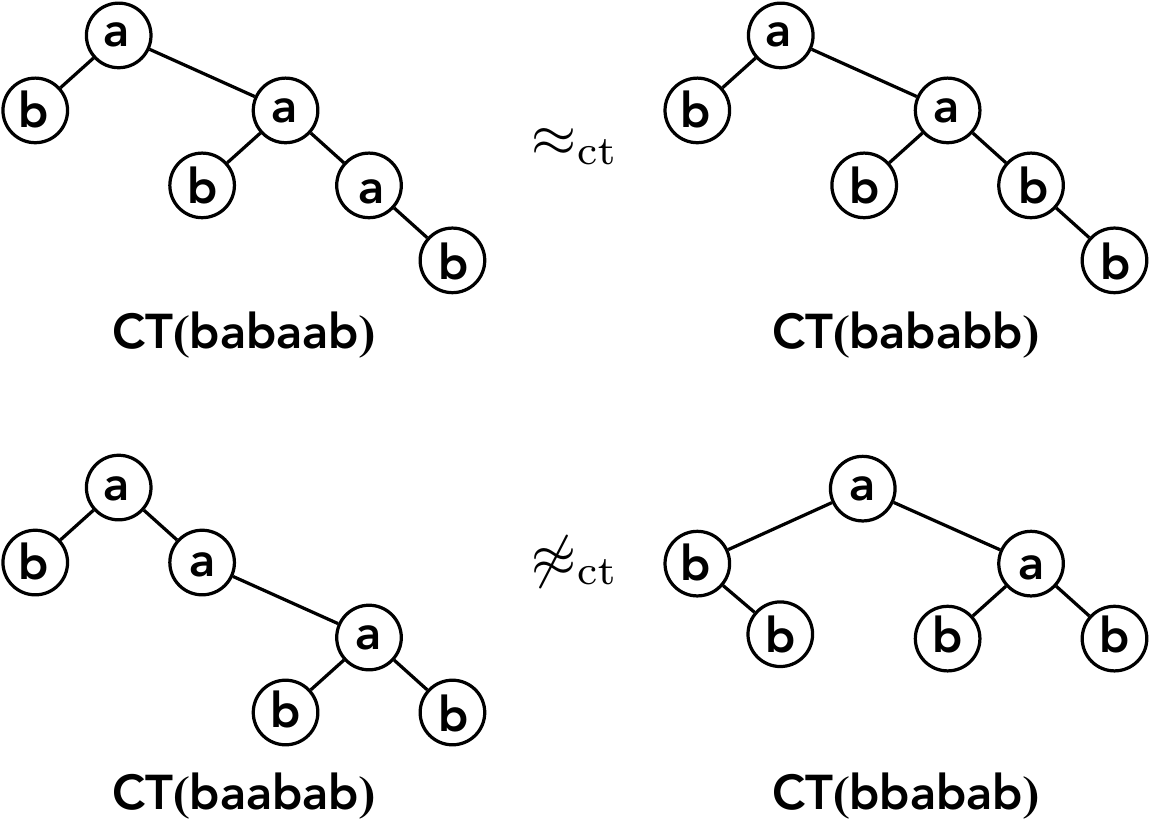}}
  \caption{Illustration for Cartesian-trees of strings $\mathtt{babaab}$, $\mathtt{bababb}$, $(\mathtt{babaab})^R$, and $(\mathtt{bababb})^R$.}
  \label{fig:reversalCT}
\end{figure}
Hence, there are cases where the Cartesian-tree sym-palindrome obtained by extending outward direction is not equal to that obtained by extending inward direction for the same center position $c$.
Therefore, we introduce the two following variants for the Cartesian-tree sym-palindromes:
\begin{definition}[outward Cartesian-tree sym-palindromes]
  A string $P$ is an outward Cartesian-tree sym-palindrome if $\rev{(P[1..\lfloor{|P|/2}\rfloor])} \approx_{\mathrm{ct}} P[\lceil{|P|/2}\rceil..|P|]$.
\end{definition}

\begin{definition}[inward Cartesian-tree sym-palindromes]
  A string $P$ is an inward Cartesian-tree sym-palindrome if
  $P[1..\lfloor{|P|/2}\rfloor] \approx_{\mathrm{ct}} \rev{(P[\lceil{|P|/2}\rceil..|P|])}$.
\end{definition}

\section{Algorithms for Computing Maximal SCSTTR Symmetry-based Palindromes}\label{sec:algo_symmetric}

This section considers algorithms for computing maximal SCSTTR sym-palindromes.
The main idea is the same as Gusfield's algorithm;
to use the outward LCE query on the matching model in SCSTTR for each center position.
Then, the complexity of the algorithm can be written as $O(nq + t_\mathrm{SCSTTR})$ time and $O(n + s_\mathrm{SCSTTR})$ space, where $q$ is the outward LCE query time, $t_\mathrm{SCSTTR}$ is the construction time of the data structure for the LCE query on the matching model, and $s_\mathrm{SCSTTR}$ is the space of the data structure.
Then, we obtain the following results for several matching models in SCSTTR.

\begin{theorem}
  All maximal SCSTTR sym-palindromes can be computed with $O(n)$ space for the following matching models:
  \begin{enumerate}
    \item For linearly sortable alphabets, all maximal Theta sym-palindromes can be computed in $O(n)$ time.
    \item For ordered alphabets, all outward/inward maximal Cartesian-tree sym-palindromes can be computed in $O(n \log n)$ time.
    \item For linearly sortable alphabets, all maximal parameterized sym-palindromes can be computed in $O(n \log (\sigma + \pi))$ time.
    \item For linearly sortable alphabets, all maximal order-preserving sym-palindromes can be computed in $O(n \log \log^2 n / \log \log \log n)$ time.
    \item For general unordered alphabets, all maximal palindromic-structure sym-palindromes can be computed in $O(n \min\{ \sqrt{\log n}, \log \sigma / \log \log \sigma \})$ time.
  \end{enumerate}
\end{theorem}

These results can be obtained by using SCSTTR suffix trees~\cite{Farach-ColtonFM00,gusfield97:_algor_strin_trees_sequen,Baker96,KimEFHIPPT14,ParkBALP20,IIT13}.
The followings are common properties among all these SCSTTR suffix trees:
\begin{enumerate}
  \item Each leaf node corresponds one-to-one to the encoding of each suffix.
  \item The string depth of the LCA of two leaves corresponds to the value of LCP of the encodings of the suffixes.
\end{enumerate}
Since these properties hold, we can compute the LCP value of encodings of two suffixes by computing the LCA. In the sequel, we show the details of each SCSTTR suffix tree.

\subsection{Maximal Theta Sym-palindromes}
An outward complementary LCE for the center position $c$ in $T$ can be computed by comparing $T[\lfloor{c+1}\rfloor..n]$ and $f(\rev{T[1..\lceil{c-1}\rceil]})$.
The above operation can be
done in constant time by constructing
the suffix tree of the string $T\$f(\rev{T})\#$ with delimiter characters $\$$ and $\#$.
\begin{lemma}\label{lem:rc_st}
  The complementary suffix tree can be constructed in $O(n)$ time and space for linearly sortable alphabets.
  Also, this data structure can answer
  an outward complementary LCE query in $O(1)$ time.
\end{lemma}

From Lemma~\ref{lem:rc_st}, we obtain the following proposition.
\begin{proposition}
  For linearly sortable alphabets, all maximal Theta sym-palindromes can be computed in $O(n)$ time with $O(n)$ space.
\end{proposition}

\subsection{Maximal Cartesian-tree Sym-palindromes}

As for a Cartesian suffix tree, the definition and the construction time are as follows:
\begin{definition}[\cite{ParkBALP20}]
  The Cartesian suffix tree of a string $T[1..n]$ is a compacted trie constructed with $\PD_{T[i..n]} \cdot \$$ for every $1\leq i \leq |T|$, where $\$ \notin \Sigma$ and $\$ \prec c$ hold for any character $c \in \Sigma$.
\end{definition}

\begin{lemma}[\cite{ParkBALP20}]\label{lem:cst_st}
  The Cartesian suffix tree of a string of length $n$ can be constructed in $O(n \log n)$ time with $O(n)$ space for ordered alphabets.
  Also, an outward/inward Cartesian-tree LCE query can be answered in $O(1)$ time by using this data structure of $T\$\rev{T}\#$.
\end{lemma}

From Lemma~\ref{lem:cst_st}, we obtain the following proposition.
\begin{proposition}
  For ordered alphabets, all outward maximal Cartesian-tree sym-palindromes can be computed in $O(n \log n)$ time with $O(n)$ space.
\end{proposition}

As for the case of inward maximal Cartesian-tree sym-palindromes, it is unclear where is the starting position of inward LCE queries.
By combining the LCE queries and the binary search, we obtain the following proposition
\begin{proposition}
  For ordered alphabets, all inward maximal Cartesian-tree sym-palindromes can be computed in $O(n \log n)$ time with $O(n)$ space.
\end{proposition}

\subsection{Maximal Parameterized Sym-palindromes}
The definition of a parameterized suffix tree is as follows:
\begin{definition}[\cite{Baker96}]
  The parameterized suffix tree of a string $T[1..n]$ is a compacted trie that stores the set $\{\mathsf{PE}_{T[i..n]} \mid 1 \leq i \leq n\}$.
\end{definition}
Therefore, by constructing the parameterized suffix tree of $T\$ \rev{T} \#$ and computing the LCP value of $\mathsf{PE}_{T[\lfloor{c+1}\rfloor..n]}$ and $\mathsf{PE}_{\rev{T[1..\lceil{c-1}\rceil]}}$, we can compute the length of the maximal parameterized palindrome for each center position $c$.

\begin{lemma}[\cite{Baker96}]\label{lem:para_st}
  The parameterized suffix tree can be constructed in $O(n \log (\sigma + \pi))$ time and $O(n)$ space for linearly sortable alphabets.
  Also, this data structure can answer
  an outward parameterized LCE query in $O(1)$ time.
\end{lemma}

From Lemma~\ref{lem:para_st}, we obtain the following proposition.
\begin{proposition}
  For linearly sortable alphabets, all maximal parameterized sym-palindromes can be computed in $O(n \log (\sigma+\pi))$ time with $O(n)$ space.
\end{proposition}

\subsection{Maximal Order-preserving Sym-palindromes}
An order-preserving suffix tree is defined as follows:
\begin{definition}[\cite{CrochemoreIKKLP16}]
  The complete order-preserving suffix tree of a string $T[1..n]$ is a compacted trie that stores $\mathsf{Code}_{T[i..n]}$ for all $1\leq i\leq n$.
  Also, the incomplete order-preserving suffix tree is a compacted trie such that some of the edge labels of the complete order-preserving suffix tree are missing.
\end{definition}

By using the (incomplete) order-preserving suffix tree, maximal order-preserving sym-{\linebreak}palindromes can be computed the same as with maximal parameterized palindrome.

\begin{lemma}[\cite{CrochemoreIKKLP16}]\label{lem:op_st}
  The (incomplete) order-preserving suffix tree can be constructed in \linebreak
  $O(n \log \log^2 n / \log \log \log n)$ time and $O(n)$ space for linearly sortable alphabets.
  Also, this data structure can answer
  an outward order-preserving LCE query in $O(1)$ time.
\end{lemma}

From Lemma~\ref{lem:op_st}, we obtain the following proposition.
\begin{proposition}
  For linearly sortable alphabets,
  all maximal order-preserving sym-palindromes can be computed in
  $O(n \log \log^2 n / \log \log \log n)$ time with $O(n)$ space.
\end{proposition}

\subsection{Maximal Palindromic-structure Sym-palindromes}
The definition of a palindromic suffix tree is as follows:
\begin{definition}
  The palindromic suffix tree of a string $T[1..n]$ is a compacted trie that represents $\LPal_{T[i..n]}$ of suffix $T[i..n]$ for all $1 \leq i \leq n$.
\end{definition}

Also, the following result is known:
\begin{lemma}[\cite{IIT13}]\label{lem:pal_st}
  The palindromic suffix tree of a string of length $n$ can be constructed in\linebreak
  $O(n \min\{ \sqrt{\log n}, \log \sigma / \log \log \sigma \})$ time
  with $O(n)$ space for general unordered alphabets, where $\sigma$ is the number of distinct characters in the string.
  Also, an outward palindromic-structure LCE query can be answered in $O(1)$ time by using this data structure of $T\$\rev{T}\#$.
\end{lemma}

From Lemma~\ref{lem:pal_st}, we obtain the following proposition.
\begin{proposition}
  For general unordered alphabets, all maximal palindromic-structure sym-palindromes can be computed in $O(n \min\{ \sqrt{\log n}, \log \sigma / \log \log \sigma \})$ time with $O(n)$ space.
\end{proposition}

\section{Algorithms for Computing Maximal SCSTTR Reversal-based\linebreak Palindromes}\label{sec:algo_reverse}

This section considers algorithms for computing maximal SCSTTR rev-palindromes.
If we were to use the SCSTTR suffix tree and outward LCE queries as in the previous section,
how to choose the starting positions of outward LCE queries is unclear.
Therefore, a na\"ive approach would require $O(n)$
inward
LCE queries for each
center position,
and the total complexity will be $O(n^2 + t_\mathrm{SCSTTR})$ time and $O(n + s_\mathrm{SCSTTR})$ space.
By combining
inward
LCE queries and binary search, we can further achieve $O(n \log n + t_\mathrm{SCSTTR})$ time and $O(n + s_\mathrm{SCSTTR})$ space with this approach.
This section shows $O(n)$-time algorithms without constructing SCSTTR suffix trees:
\begin{theorem}
  There exist $O(n)$-time algorithms which compute
  \begin{enumerate}
    \item all maximal Theta rev-palindromes for general unordered alphabets,
    \item all maximal Cartesian-tree rev-palindromes for linearly sortable alphabets,
    \item all maximal parameterized rev-palindromes for linearly sortable alphabets,
    \item all maximal order-preserving rev-palindromes for general ordered alphabets, and
    \item all maximal palindromic-structure rev-palindromes for general unordered alphabets.
  \end{enumerate}
\end{theorem}

\paragraph*{Outline of Computing SCSTTR Reversal-based Palindromes}
At first, we show a framework for computing maximal SCSTTR rev-palindromes, which is a generalization of Manacher's algorithm~\cite{Manacher75}.
For the sake of simplicity, we denote SCSTTR rev-palindromes by just \emph{palindromes} in the description of the framework below.
In the following, we describe how to compute all odd-lengthed maximal palindromes.
Even-lengthed maximal palindromes can be obtained analogously.

We consider finding the odd-lengthed maximal palindromes in ascending order of the center position.
Let $\MPal[c]$ denote the interval corresponding to the odd-lengthed maximal palindrome centered at position $c$.
The length of $\MPal[1]$ is always one.
Assuming that all odd-lengthed maximal palindromes whose center is at most $c-1$ have been computed and let $\MPal[c'] = [b, e]$ be the longest interval whose ending position $e$ is largest among $\{\MPal[1],\ldots,\MPal[c-1]\}$.
Further, let $P' = T[b.. e]$.
Also, let $d$ be the distance between $c'$ and $c$; namely, $d = c - c'$.
See also Figure~\ref{fig:manacher} for illustration.
\begin{figure}[htb]
  \centering
  \includegraphics[width=0.8\linewidth]{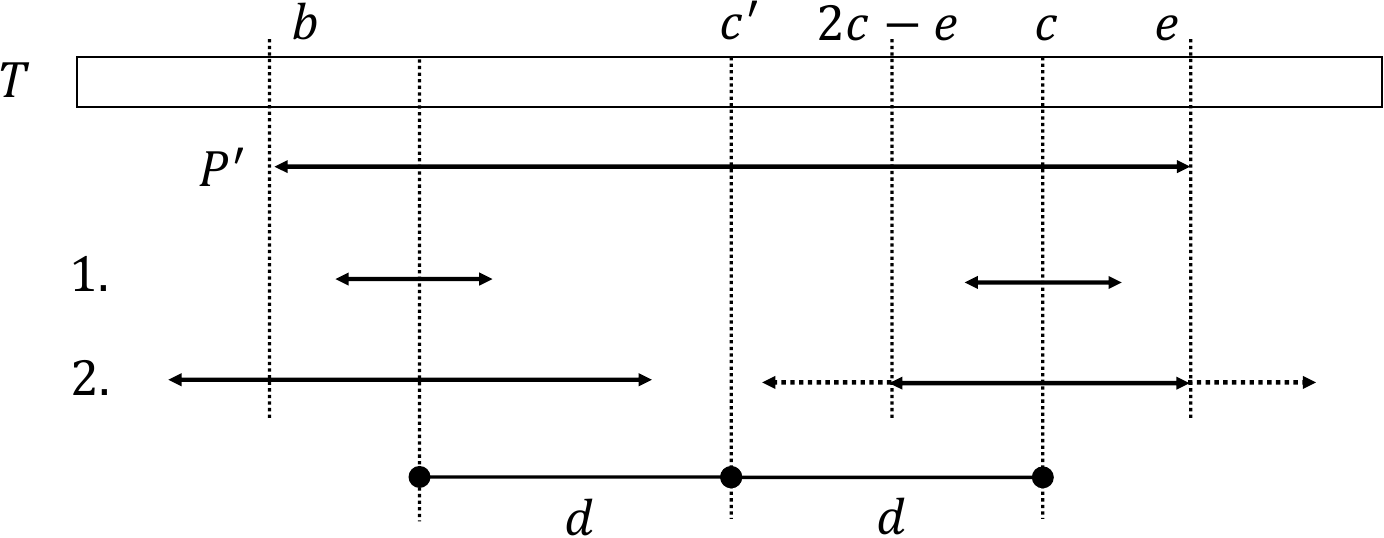}
  \caption{
  Illustration for two sub-cases in $c \leq e$.}
  \label{fig:manacher}
\end{figure}

If $T[i-1..j+1]$ is a palindrome, we say palindrome $T[i..j]$ can be extended.
Also, we call computing the length of $\MPal[(i+j)/2]$ computing the extension of $T[i..j]$.

Now, we consider how to compute $\MPal[c]$.
If $e < c$, then $c = e + 1$ holds.
We then compute $\MPal[c]$ by computing the extension of $T[e+1..e+1]$.
Otherwise (i.e., if $c \leq e$ holds), we compute $\MPal[c]$ as in the following two cases according to the relationship between $P'$ and $\MPal[c'-d]$.
\begin{enumerate}
  \item
    If the starting position of $\MPal[c'-d]$ is larger than $b$, then $|\MPal[c]| = |\MPal[c'-d]|$ holds (by Lemma~\ref{lem:symmetricity_of_SCSTTR_revpal}).
    Thus we \emph{copy} $\MPal[c'-d]$ to $\MPal[c]$ with considering symmetry.
  \item Otherwise, the ending position of $\MPal[c]$ is at least $e$.
    Then we compute the extension of $T[2c-e..e]$, and $P'$ is updated to $\MPal[c]$ if $T[2c-e..e]$ can be extended.
\end{enumerate}

Note that, in the exact matching model, this framework is identical to that of Manacher's algorithm.
The framework includes two non-trivial operations; copying and extension-computing.
How many times both operations are called can be analyzed in the same manner as for Manacher's algorithm, and is $O(n)$.
Also, each copying operation can be done in constant time.
Thus, if we can perform each extension-computing in (amortized) constant time, the total time complexity becomes $O(n)$.
The following sections focus only on how to compute an extension in (amortized) constant time
under each of
the two matching models; the complementary-matching model and the Cartesian-tree matching model.

\subsection{Maximal Theta Reversal-based Palindromes}

In the complementary-matching model, the correspondence of matching characters, although different, is predetermined, and na\"ive character comparisons can be done in constant time.
Thus, the following result can be obtained immediately.

\begin{proposition}
  For general unordered alphabets, all maximal Theta rev-palindromes can be computed in $O(n)$ time with $O(n)$ space.
\end{proposition}
\subsection{Maximal Cartesian-tree Reversal-based Palindromes}

Here, we consider how to compute the extension of a maximal Cartesian-tree rev-palindrome $T[i..j]$.
We consider the difference between $\PD_{T[i..j]}$ and $\PD_{T[i-1..j+1]}$.
There are three types of positions of
$\PD_{T[i-1..j+1]}$ in which values differ from $\PD_{T[i..j]}$;
the first and last positions (which do not exist in $\PD_{T[i..j]}$),
and each position $k$ such that $\PD_{T[i..j]}[k] = 0$ and $T[i-1] \preccurlyeq T[k]$ with $i \leq k \leq j$ (i.e., $\PD_{T[i-1..j+1]}[k+1] \neq 0$).
Let $m_{[i..j]}$ be the leftmost occurrence position of the smallest value in $T[i..j]$.
Then, $k'$ such that $\PD_{T[i..j]}[k'] = 0$ is always to the left or equal to $m_{[i..j]}$.
However, since the number of such positions is not always constant, the time required for computing the extension from $T[i..j]$ to $T[i-1..j+1]$ can be $\omega(1)$.
For example, when $T[i-1.. j+1] = \mathtt{becaebdaefc}$,
$\PD_{T[i..j]} = [\underline{0}, \underline{0}, 0, 1, 2, 1, 4, 1, 1]$ and $\PD_{T[i-1.. j+1]} = [0,\underline{1},\underline{2},0,1,2,1,4,1,1,3]$ hold.
There are two differences between $\PD_{T[i..j]}$ and $\PD_{T[i-1.. j+1]}[2..10]$ (underlined).
We will consider the total number of updates of such positions through the entire algorithm.

Let $c$ be the center of $T[i..j]$ and $T[b..e]$ be the maximal Cartesian-tree rev-palindrome centered at $c$.
Also, let $Z_c$ be the number of $0$'s in $\PD_{T[i..j]}$, $E_c$ be the length of the extension of $T[i..j]$ (i.e., $E_c = i-b$), and $U_c$ be the number of updates of positions of $0$ from $\PD_{T[i..j]}$ to $\PD_{T[b..e]}$.
Let $T[i'..j']$ be the next palindrome of center $c'$ which we have to compute the extension of, after computing the extension of $T[i..j]$.
Namely, $T[i'..j']$ is the longest proper suffix palindrome of $T[b..e]$ if $|T[b..e]| > 1$.
Otherwise, $T[i'..j'] = T[e+1..e+1]$ holds.

As the first step to considering the total sum $\sum{U_c}$ over the entire algorithm, we consider the simpler case such that each character is distinct from the other in the string.
We call this case \emph{permutation-Cartesian-tree matching}.

In this case, $c=m_{[i..j]}$ holds since $T[m_{[i..j]}]$ is the only smallest value in $T[i..j]$ and $m_{[i..j]}$ is the root of $\CT(T[i..j])$ and $\CT(\rev{T[i..j]})$.
Also, $m_{[i..j]} = m_{[b..e]}$ holds.
Now we consider the starting position of the next palindrome $T[i'..j']$.

\begin{lemma}\label{lem:pct}
  Let $T[i'..j']$ be the next palindrome. Then, $i' > m_{[i..j]}$ holds.
\end{lemma}

\begin{proof}
  The next palindrome $T[i'..j']$ is either $T[e+1..e+1]$ or the longest proper suffix palindrome of $T[b..e]$.
  Since the former case is obvious, we consider only the latter in the following.
  For the sake of contradiction, we assume that $i' \leq m_{[i..j]}$.
  Then $T[m_{[i'..j']}] < T[m_{[i..j]}]$ holds.
  Since $m_{[i'..j']} \in [b..e]$, this leads to a contradiction with $m_{[i..j]} = m_{[b..e]}$.
\end{proof}

From Lemma~\ref{lem:pct}, $0$s in $\PD_{T[i..j]}$ and $0$s in $\PD_{T[i'..j']}$ all correspond to different positions in $T$.
Hence, $\sum{Z_c} \leq n$ holds.
Also, $U_c \leq Z_c + E_c$ and $\sum{E_c} \leq n$ holds.
Therefore, we obtain $\sum{U_c} \leq \sum{Z_c} + \sum{E_c} \leq 2n$.
In the preprocessing,
we compute $\PD_T$ and $\PD_{\rev{T}}$ in $O(n)$ time~\cite{ParkBALP20}
and create a link from each position $p$ to $q$
with $\max_{1 \leq q < p} \{ q : T[q] \preccurlyeq T[p]\}$.
Then, when we compute the extension of $T[i..j]$,
each $0$'s position to be updated can be obtained in $O(1)$ time
by traversing back the links from $i-1$.
Thus, the extension from $T[i..j]$ to $T[b..e]$ can be computed in $O(E_c + U_c)$ time by updating $0$s
and checking the equality's of values of all updated positions between $\PD_{T[i-\ell..j+\ell]}$ and $\PD_{\rev{T[i-\ell..j+\ell]}}$
for each $1 \leq \ell \leq E_c$,
since we can access each element of $\PD_{T[i-\ell..j+\ell]}$ and $\PD_{\rev{T[i-\ell..j+\ell]}}$ in constant time if we have $\PD_T$ and $\PD_{\rev{T}}$~\cite{ParkBALP20}.
Therefore, we obtain the following result:
\begin{lemma}
  For linearly sortable alphabets, all maximal permutation-Cartesian-tree rev-palindromes can be computed in $O(n)$ time with $O(n)$ space.
\end{lemma}

Now we consider the sum $\sum{U_c}$ in the general case.
As for the relationship between $Z_c$ and $Z_{c'}$, we show the following lemma:

\begin{lemma}\label{lem:ct}
  For the center $c$ of $T[i.. j]$ and the center $c'$ of the next palindrome $T[i'.. j']$,
  $Z_{c'} \leq Z_c + E_c - U_c$ holds.
\end{lemma}

\begin{proof}
  The next palindrome $T[i'..j']$ is either $T[e+1..e+1]$ or the longest proper suffix palindrome of $T[b..e]$.
  Since the former case is obvious, we consider only the latter case.
  By definitions, the number of $0$s in $\PD_{T[b..e]}$ is equal to $Z_c + E_c - U_c$.
  Also, since ${T[b..e]}$ is a Cartesian-tree palindrome, the number of $0$s in $\PD_{\rev{T[b..e]}}$ is also equal to $Z_c + E_c - U_c$.
  Similarly,
  since $T[i'..j']$ is a Cartesian-tree palindrome, the number of $0$s in $\PD_{\rev{T[i'..j']}}$ is equal to $Z_{c'}$.
  Further, the number of $0$s in $\PD_{\rev{T[b..e]}}$ is at least $Z_{c'}$
  since $\rev{T[i'..j']}$ is a prefix of $\rev{T[b.. e]}$~(see also Figure~\ref{fig:ct_example}).
  Therefore  $Z_{c'} \leq Z_c + E_c - U_c$ holds.
\end{proof}
\begin{figure}[tb]
  \centerline{
    \includegraphics[width=0.8\linewidth]{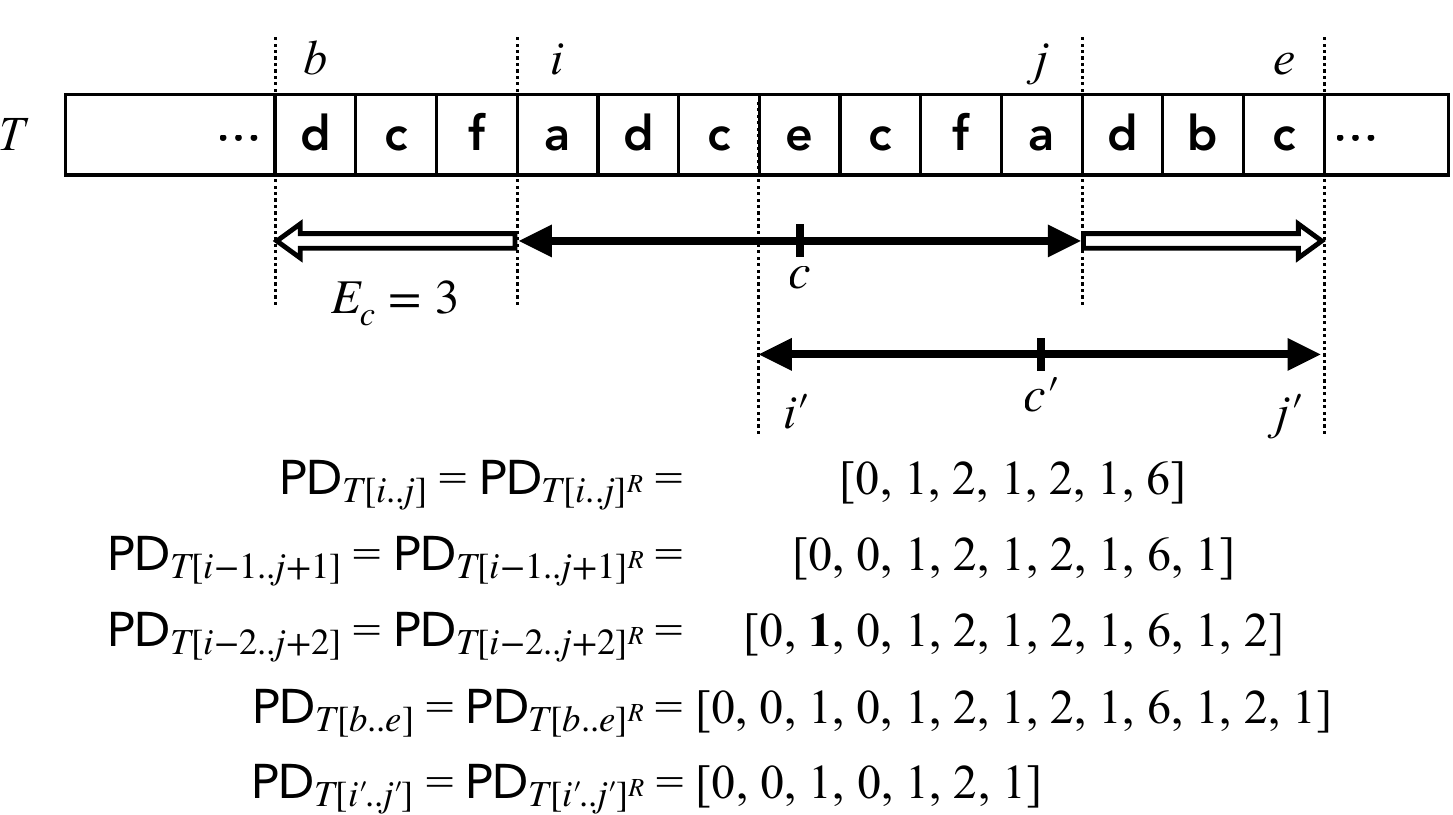}
  }
  \caption{A concrete example for Cartesian-tree palindromes. In this example, $Z_c = 1$, $E_c = 3$, and $U_c = 1$ hold
    since when $\PD_{T[i-1.. j+1]}$ grows to $\PD_{T[i-2.. j+2]}$,
  a single value (shown in bold font) is updated from $0$ to non-zero.}
  \label{fig:ct_example}
\end{figure}
From Lemma~\ref{lem:ct}, $U_c \leq Z_c - Z_{c'} +E_c$ holds.
Therefore, the total number of updates $\sum{U_c} \leq \sum(Z_c - Z_{c'}) + \sum{E_c} \leq 2n$.
Hence, we obtain the following result:

\begin{proposition}
  For linearly sortable alphabets, all maximal Cartesian-tree rev-palindromes can be computed in $O(n)$ time with $O(n)$ space.
\end{proposition}

\subsection{Maximal Parameterized Reversal-based Palindromes}
Now we consider determining whether the parameterized palindrome $T[i..j]$ can be extended by one character at each end.
If $T[i-1..j+1] \approx_{\mathrm{para}} \rev{T[i-1..j+1]}$ holds, then $\mathsf{PE}_{T[i-1..j+1]} = \mathsf{PE}_{\rev{T[i-1..j+1]}}$.
Since $T[i..j]$ is a parameterized palindrome, $\mathsf{PE}_{T[i..j]} = \mathsf{PE}_{\rev{T[i..j]}}$.
Now, it is not difficult to see
that the number of positions in which values of
$\mathsf{PE}_{T[i-1..j+1]}$ can differ from $\mathsf{PE}_{T[i..j]}$
when aligned at the center, is at most three;
the first and last elements of $\mathsf{PE}_{T[i-1..j+1]}$ (which do not exist in $\mathsf{PE}_{T[i..j]}$),
and the position of the first element of $T[i..j]$ that equals $T[i-1]$, if such exists.
From the symmetricity, the same can be said for values
of $\mathsf{PE}_{\rev{T[i-1..j+1]}}$ from $\mathsf{PE}_{\rev{T[i..j]}}$.
Then, to determine whether $T[i-1..j+1]$ is a parameterized palindrome, we only have to consider
(1) the equality between the first elements of $\mathsf{PE}_{T[i-1..j+1]}$ and $\mathsf{PE}_{\rev{T[i-1..j+1]}}$,
(2) the equality between the last elements of $\mathsf{PE}_{T[i-1..j+1]}$ and $\mathsf{PE}_{\rev{T[i-1..j+1]}}$, and
(3) the equality between the first positions $p$ in $T[i..j]$ and $q$ in $\rev{T[i..j]}$
that respectively equals $T[i-1]$ and $T[j+1]$.
In other words, given that $T[i..j]$ is a parameterized palindrome,
$T[i-1..j+1]$ is a parameterized palindrome if and only if all three equality's hold.
See also Table~\ref{table:para}.
\begin{table}[h]
  \centering
  \begin{tabular}{|c||c|c|c|c|c|c|c|c|c|c|c|c|c|c|c|}
    \hline
    $k$                                         & $\cdots$ & $i-1$        &              &              &              &              &              &              &              &              &              & $j+1$        & $\cdots$ \\ \hline
    $T[k]$                                      & $\cdots$ & $\mathsf{c}$ & $\mathsf{a}$ & $\mathsf{a}$ & $\mathsf{c}$ & $\mathsf{a}$ & $\mathsf{e}$ & $\mathsf{b}$ & $\mathsf{d}$ & $\mathsf{b}$ & $\mathsf{b}$ & $\mathsf{d}$ & $\cdots$ \\ \hline
    $k'$                                        &          & $1$          &              &              & $p$          &              &              &              &              &              &              & $j-i+3$      &          \\ \hline
    $\mathsf{PE}_{T[i-1..j+1]}[k']$             &          & $0$          & $0$          & $1$          & $3$          & $2$          & $0$          & $0$          & $0$          & $2$          & $1$          & $3$          &          \\ \hline\hline
    $\hat{k}$                                   & $\cdots$ & $n-j$        &              &              &              &              &              &              &              &              &              & $n-i+2$      & $\cdots$ \\ \hline
    $\rev{T}[\hat{k}]$                          & $\cdots$ & $\mathsf{d}$ & $\mathsf{b}$ & $\mathsf{b}$ & $\mathsf{d}$ & $\mathsf{b}$ & $\mathsf{e}$ & $\mathsf{a}$ & $\mathsf{c}$ & $\mathsf{a}$ & $\mathsf{a}$ & $\mathsf{c}$ & $\cdots$ \\ \hline
    $\hat{k'}$                                  &          & $1$          &              &              & $q$          &              &              &              &              &              &              & $j-i+3$      &          \\ \hline
    $\mathsf{PE}_{\rev{T[i-1..j+1]}}[\hat{k'}]$ &          & $0$          & $0$          & $1$          & $3$          & $2$          & $0$          & $0$          & $0$          & $2$          & $1$          & $3$          &          \\ \hline
  \end{tabular}
  \caption{An example for the computation of whether the parameterized palindrome $T[i..j]$ can be extended.
  }
  \label{table:para}
\end{table}

The equality of (1) always holds, since $\mathsf{PE}_{T[i-1..j+1]}[1]$ and $\mathsf{PE}_{\rev{T[i-1..j+1]}}[1]$ are
always $0$ by definition of $\mathsf{PE}$.
Next, we show that the equality of (2) implies equality of (3).
Since the last elements of $\mathsf{PE}_{T[i-1..j+1]}$ and $\mathsf{PE}_{\rev{T[i-1..j+1]}}$ are equal,
this implies that the last position in $T[i..j]$ and $\rev{T[i..j]}$ that respectively equals
$T[j+1]$ and $T[i-1]$ are the same.
By reversing the string,
these correspond to the first position in $\rev{T[i..j]}$ and $T[i..j]$ that respectively equals
$T[j+1]$ and $T[i-1]$ therefore implying (3).

Thus, we need only check the equality of (2).
Given $\mathsf{PE}_T$, $\mathsf{PE}_{T[i..j]}[k]$ for any $1 \leq k \leq j - i + 1$
can be computed in constant time, i.e.,
$\mathsf{PE}_{T[i..j]}[k] = \mathsf{PE}_{T}[i+k-1]$ if $\mathsf{PE}_{T}[i+k-1] < k$ and $0$ otherwise.
Hence, the computation time of the extension of $T[i..j]$ by one character at each end is $O(1)$.
Also, since $\mathsf{PE}_T$ can be precomputed in $O(n)$ time for linearly sortable alphabets, we obtain the following result:
\begin{proposition}
  For linearly sortable alphabets, all maximal parameterized rev-palindromes can be computed in $O(n)$ time with $O(n)$ space.
\end{proposition}

\subsection{Maximal Order-preserving Reversal-based Palindromes}
From the definition, an order-preserving rev-palindrome $P$ satisfies that
$P[\lceil{c\rceil}-d] \preccurlyeq P[\lfloor{c\rfloor}+d]$ and $P[\lfloor{c\rfloor}+d] \preccurlyeq P[\lceil{c\rceil}-d]$
for integer $d$ with $1 \leq d \leq \lceil{c\rceil}-1$
where $c = \frac{1+|P|}{2}$ is the center of $P$.
This means that any order-preserving rev-palindrome must be a palindrome under the exact matching.
Therefore, maximal order-preserving rev-palindromes in $T$ can be computed in $O(n)$ time by using Manacher's algorithm.
\begin{corollary}
  For general ordered alphabets, all maximal order-preserving rev-palindromes can be computed in $O(n)$ time with $O(n)$ space.
\end{corollary}

\subsection{Maximal Palindromic-structure Reversal-based Palindromes}
Now we consider how to determine whether the palindromic-structure palindrome $T[i..j]$ can be extended by one character at each end.
The changes from the set of maximal palindromes in $T[i..j]$ to that of $T[i-1..j+1]$ are prefix palindromes and suffix palindromes of $T[i-1..j+1]$.
Therefore, if the set of prefix palindromes is equal to the set of suffix palindromes, then $T[i-1..j+1]$ is a palindromic-structure palindrome.
To efficiently compute whether this condition is satisfied, we rewrite the condition as follows:
\begin{enumerate}
  \item Among maximal palindromes in $T[1..n]$, the set of maximal palindromes starting at $i$ whose length is at most $j-i$ is equal to the set of maximal palindromes ending at $j$ whose length is at most $j-i$.
  \item The presence or absence of ``a prefix palindrome of length two beginning at position $i-1$'' and ``a suffix palindrome of length two ending at position $j+1$'' correspond.
\end{enumerate}
Since a suffix palindrome of $j+1$ is either a suffix palindrome of $j+1$ whose length is at most two or an extension of a suffix palindrome of $j$ which is not a maximal palindrome ending at $j$, the above rewrites hold.
See also Figure~\ref{fig:pal}.
\begin{figure}[ht]
  \centerline{
    \includegraphics[width=0.8\linewidth]{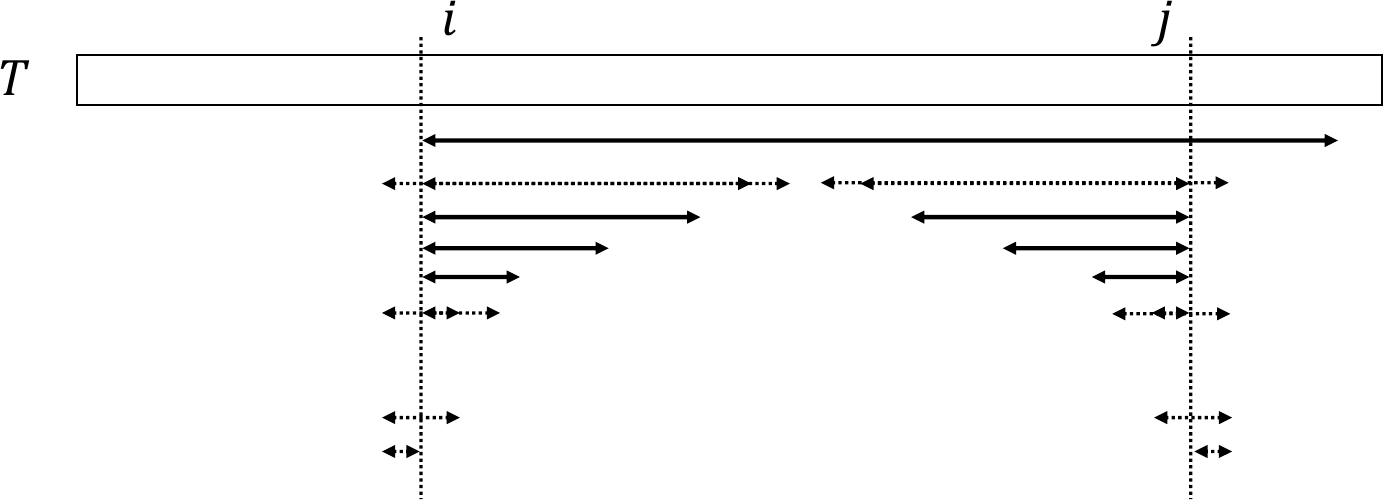}
  }
  \caption{Illustration for palindromic structures.
    Solid arrows are maximal palindromes starting at $i$ or ending at $j$, and
    dotted arrows are prefix or suffix palindromes.
    The set of suffix palindromes of $T[i-1..j+1]$ is a union of the set of the extensions of suffix palindromes of $T[i..j]$ and the set of suffix palindromes of lengths one and two.
  The prefix/suffix palindrome of length one always exists.}
  \label{fig:pal}
\end{figure}
Condition 2 can be computed easily; hence, we show how to compute condition 1 below.
First, we precompute all maximal palindromes in $T[1..n]$ by using Manacher's algorithm.
Then, we construct an ascending sequence of lengths of maximal palindromes with the same ending (starting) position, concatenated with a delimiter character $\$_i$, where $i$ is the ending (starting) position.
For example, the sequence for the ending position for a string $\mathsf{abaaa}$ is $1 \$_1 \$_2 1 3 \$_3 2 \$_4 1 2 3 \$_5$.
By using these sequences for starting/ending position and LCE data structure, condition 1 can be computed in constant time.
The length of the sequence is $O(n)$ since the number of maximal palindromes is $2n-1$.
Then we obtain the following proposition

\begin{proposition}
  For general unordered alphabets, all maximal palindromic-structure rev-palindromes can be computed in $O(n)$ time with $O(n)$ space.
\end{proposition}

Note that the concept of the palindromic structure can be generalized from maximal palindromes to maximal SCSTTR palindromes.
From the above results, if maximal SCSTTR palindromes are given, maximal SCSTTR palindromic-structure rev-palindromes can be computed in $O(n)$ time and space.

\section{Conclusions and Future Work}\label{sec:conclusion}
In this paper, we dealt with the problems of computing all maximal palindromes in a string under several
variants of string matching models.
We first showed two distinct definitions of palindromes under the non-standard matching models: symmetric and reverse.
For maximal sym-palindromes, we proposed suffix-tree-based algorithms that work in their construction time and $O(n)$ space.
Also, for maximal rev-palindromes, we showed that
Manacher's algorithm can be generalized to the matching models we consider and presented $O(n)$ time and space algorithms.

Our future work includes the following:
\begin{enumerate}
  \item Can we efficiently compute maximal palindromes under other string matching models?
    For example, can we extend our algorithm to a structural matching model~\cite{Shibuya04}, which is a generalization of a parameterized matching?
  \item It is known that any string of length $n$ can contain at most $n+1$ distinct palindromes~\cite{DroubayJP01}, and all distinct palindromes can be computed in $O(n)$ time~\cite{GroultPR10}.
    It is interesting to study distinct palindromes under SCSTTRs.
    Rubinchik and Shur~\cite{RubinchikS18} claimed that their \emph{eertree} data structure for computing distinct palindromes can be extended to computing distinct WK-palindromes.
    We are unaware of whether the computed palindromes are symmetric-based or reversal-based.
  \item While there can be $\Omega(n)$ suffix palindromes in a string of length $n$, their lengths can be compactly represented with $O(\log n)$ arithmetic progressions~\cite{ApostolicoBG95,matsubara_tcs2009}. Can suffix sym-/rev-palindromes under SCSTTRs be represented compactly in a similar way?
\end{enumerate}

\section*{Acknowledgments}
This work was supported by JSPS KAKENHI Grant Numbers JP23H04381 (TM),
JP20J21147 (MF), JP21K17705 (YN), JP23H04386 (YN),
JP20H05964 (SI), JP23K24808 (SI), JP24K02899 (HB), and JP18H04098 (MT).

\end{document}